\pdfoutput=1 

\documentclass[a4paper,USenglish,cleveref, autoref]{lipics-v2019}

\hideLIPIcs 

\title{Fast Multi-Subset Transform and Weighted Sums Over Acyclic Digraphs} 

\titlerunning{Fast Multi-Subset Transform}

\author{Mikko Koivisto}{Department of Computer Science, University of Helsinki, Finland}{mikko.koivisto@helsinki.fi}{}{}

\author{Antti R\"{o}ysk\"{o}}{Department of Computer Science, University of Helsinki, Finland}{antti.roysko@helsinki.fi}{}{}

\authorrunning{M. Koivisto and A. R\"{o}ysk\"{o}}

\Copyright{Mikko Koivisto and Antti R\"{o}ysk\"{o}}



\category{}

\relatedversion{}

\supplement{}

\funding{This work was partially supported by the Academy of Finland, Grant 316771.}

\acknowledgements{We thank Petteri Kaski for valuable discussions about the topic of the paper.}

\nolinenumbers 

\EventEditors{Susanne Albers}
\EventNoEds{1}
\EventLongTitle{17th Scandinavian Symposium and Workshops on Algorithm Theory (SWAT 2020)}
\EventShortTitle{SWAT 2020}
\EventAcronym{SWAT}
\EventYear{2020}
\EventDate{June 22--24, 2020}
\EventLocation{T\'{o}rshavn, Faroe Islands}
\EventLogo{}
\SeriesVolume{162}
\ArticleNo{28}

\usepackage{clrscode3e}

\theoremstyle{plain}
\newtheorem{fact}[theorem]{Fact}
\newtheorem{conjecture}[theorem]{Conjecture}

\newcommand{\rmm}{\omega}
\newcommand{\cC}{{\mathscr{C}}}
\newcommand{\cK}{{\mathscr{K}}}
\newcommand{\cR}{{\mathscr{R}}}
\newcommand{\cS}{{\mathscr{S}}}
\newcommand{\cT}{{\mathscr{T}}}

\newcommand{\be}{\begin{eqnarray}}
\newcommand{\ee}{\end{eqnarray}}
\newcommand{\bes}{\begin{eqnarray*}}
\newcommand{\ees}{\end{eqnarray*}}

\long\def\comment#1{}

\renewcommand{\gets}{\leftarrow}

\begin{document}

\maketitle

\begin{abstract}
The zeta and Moebius transforms over the subset lattice of $n$ elements and the so-called subset convolution are examples of unary and binary operations on set functions. While their direct computation requires $O(3^n)$ arithmetic operations, less naive algorithms only use $2^n \mathrm{poly}(n)$ operations, nearly linear in the input size.  
Here, we investigate a related $n$-ary operation that takes $n$ set functions as input and maps them to a new set function. This operation, we call multi-subset transform, is the core ingredient in the known inclusion--exclusion recurrence for weighted sums over acyclic digraphs, which extends Robinson's recurrence for the number of labelled acyclic digraphs. Prior to this work, the best known complexity bound for computing the multi-subset transform was the direct $O(3^n)$. By reducing the task to rectangular matrix multiplication, we improve the complexity to $O(2.985^n)$. 
\end{abstract}

\section{Introduction}

In this paper, we consider the following problem. We are given a finite set $U$ and, for each element $i \in U$, a function $f_i$ from the subsets of $U$ to some ring $\cR$. The task is to compute the function $g$ given by 
\be \label{eq:mst}
    g(T) = \sum_{S \subseteq T} \prod_{i \in T} f_i(S)\,,\quad T \subseteq U\,.
\ee
We shall call $g$ the \emph{multi-subset transform} of $(f_i)_{i\in U}$. While the present study of this operation on set functions stems from a particular application to weighted counting of acyclic digraphs, which we will introduce later in this section, we believe the multi-subset transform could also have applications elsewhere.  

A straightforward computation of the multi-subset transform requires $\Omega(3^n)$ arithmetic operations (i.e., additions and multiplications in the ring $\cR$) when $U$ has $n$ elements. In the light of the input size $O(2^n n)$ and output size $O(2^n)$, one could hope for an algorithm that requires $2^n n^{O(1)}$ operations.  
Some support for optimism is provided by the close relation to two similar operations on set functions: the zeta transform of $f$ and the subset convolution of $f_1$ and $f_2$, given respectively by 
\bes 
	(f \zeta)(T) = \sum_{S \subseteq T} f(S)\quad \textrm{and} \quad
	(f_1 * f_2)(T) = \sum_{S \subseteq T} f_1(S) f_2(T\!\setminus\! S)\,,
	\quad T \subseteq U\,; 
\ees
these unary and binary operations can be performed using $O(2^n n)$ \cite{Yates1937,Kennes1992} and $O(2^n n^2)$ \cite{Bjorklund2007} arithmetic operations, thus significantly beating the naive $\Omega(3^n)$-computation. 
Indeed, consider the seemingly innocent replacement of ``$i \in T$'' by ``$i \in S$'' or ``$i \in U$'' in \eqref{eq:mst}: either one would yield a variant that immediately (and efficiently) reduces to the zeta transform. Likewise, replacing the factor $\prod_{i\in T\setminus S} f_i(S)$ in the product by $\prod_{i\in T\setminus S} f_i(T\!\setminus\! S)$ would give us an instance of subset convolution. The present authors do not see how to fix these ``broken reductions''---the multi-subset transform could be a substantially harder problem not admitting a nearly linear-time algorithm. One might even be tempted to hypothesize that one cannot reduce the base of the exponential complexity below the constant $3$. We refute this hypothesis: 

\begin{theorem}\label{thm:fmst}
The multi-subset transform can be computed using $O(2.985^n)$ arithmetic operations.
\end{theorem}

We obtain our result by a reduction to \emph{rectangular matrix multiplication (RMM)}. The basic idea is to split the ground set $U$ into two halves $U_1$ and $U_2$ and divide the product over $i \in T$ into two smaller products accordingly. In this way we can view \eqref{eq:mst} as a matrix product of dimensions $2^{|U_1|} \times 2^{|U|} \times 2^{|U_2|}$. The two rectangular matrices are sparse, with at most $6^{n/2} = O(2.4495^n)$ non-zero elements out of the total $8^{n/2}$. The challenge is to exploit the sparsity. Known algorithms for general sparse matrix multiplication \cite{Yuster2005,Kaplan2006} turn out to be insufficient for getting beyond the $O(3^n)$ bound (see Section~\ref{se:basic} for details). Fortunately, in our case the sparsity occurs in a special, structured form that enables better control of zero-entries, and thereby a more efficient reduction to dense RMM. To get the best available constant base in the exponential bound, we call upon the recently improved fast RMM algorithms \cite{Gall2018}.

\subsection{Application to weighted counting of acyclic digraphs}

Let $a_n$ be the number of labeled acyclic digraphs on $n$ nodes. Robinson \cite{Robinson1973} and Harary and Palmer \cite{Harary1973}, independently discovered the following inclusion--exclusion recurrence: 
\bes
	a_n = \sum_{s=1}^{n} (-1)^{s-1} {n \choose s} 2^{s(n-s)} a_{n-s}\,.
\ees
To see why the formula holds, view $s$ as the number of sinks (i.e., nodes with no out-neighbors), each of which can choose its in-neighbors freely form the remaining $n-s$ nodes.  

Tian and He \cite{Tian2009} generalized the recurrence to weighted counting of acyclic digraphs on a given set of $n$ nodes $V$. Now every acyclic digraph $D$ on $V$ is assigned a \emph{modular} weight, that is, a real-valued weight $w(D)$ that factorizes into node-wise weights $w_i(D_i)$, where $D_i \subseteq V\!\setminus\!\{i\}$ is the set of in-neighbors of node $i$ in $D$. This counting problem has applications particularly in Bayesian learning of Bayesian networks from data; the weighted count is the partition function of a statistical model that associates each node of the graph with a random variable, and  evaluating the partition function is the main computational bottleneck  \cite{Friedman2003,Tian2009,Talvitie2019}. Letting $a_V$ denote the weighted sum of acyclic digraphs on $V$, we have 
\be\label{eq:aV}
	a_V = \sum_{D} \prod_{i\in V} w_i(D_i) = \sum_{\emptyset \neq S \subseteq V} (-1)^{|S|-1} \Bigg(\prod_{i\in S} \sum_{D_i \subseteq V \setminus S} w_i(D_i)\Bigg)\, a_{V \setminus S}\,.
\ee
The recurrence enables computing $a_V$ using $O(3^{n} n)$ arithmetic operations \cite{Tian2009}.

We will apply Theorem~\ref{thm:fmst} to lower the base of the exponential bound: 

\begin{theorem}\label{thm:application}
The sum over acyclic digraphs with modular weights can be computed  using $O(2.985^n)$ arithmetic operations.
\end{theorem}

\subsection{Related work}

There are numerous previous applications of fast matrix multiplication algorithms to decision, optimization, and counting problems. Here we only mention a few that are most related to the present work. 

Williams \cite{Williams2005} employs fast \emph{square} matrix multiplication to count all variable assignments that satisfy a given number of constraints, each involving at most two variables. By a simple reduction, this yields the fastest known algorithm for the Max-2-CSP problem. The present work is based on the same idea of viewing the product of a group of low-arity functions as a large matrix; this general idea is also studied in the doctoral thesis of the first author \cite[Sects.~3.3 and 3.6]{Koivisto2004}, including reductions to \emph{RMM}, however, without concrete applications.   

Bj{\"{o}}rklund, Kaski, and Kowalik \cite{Bjorklund2017} apply fast RMM to show the following: Given a nonnegative integer $q$ and three mappings $f$, $g$, $h$ from the subsets of an $n$-element set to some ring, one can sum up the products $f(A)g(B)h(C)$ over all pairwise disjoint triplets of $q$-sets $A, B, C$ using $O\big(n^{3 q \tau + c}\big)$ ring operations, where $\tau < \tfrac{1}{2}$ and $c \geq 0$ are constants independent of $q$ and $n$. Consequently, one can count the occurrences of constant-size paths (or any other small-pathwidth patterns) faster than in the ``meet-in-the-middle time'' \cite{Bjorklund2017}. While the involvement of set functions and set relations bear a resemblance to those in multi-subset transform, the reduction of  Bj{\"{o}}rklund et al.\ is based on solving an appropriately  constructed system of linear equations, and is thus very different from the combinatorial approach taken in the present work. 

\section{Fast multi-subset transform: proof of Theorem~1}
\label{se:proof1}

We will develop an algorithm for multi-subset transform in several steps. In Section~\ref{se:basic} we give the basic reduction to RMM and the idea of splitting the sum over into several smaller sums. Then, in Section~\ref{se:first} we present a simple implementation of the splitting idea, and get our first below-$3$ algorithm. This algorithm is improved upon in Section~\ref{se:second}, yielding the claimed complexity bound. We end this section by presenting a more sophisticated splitting scheme in Section~\ref{se:third}. We have not succeeded to give a satisfactory analysis of its complexity. Yet, our numerical calculations suggest the bound $O(2.930^n)$.  

We will denote by $\rmm(k)$, for $k \geq 0$, the smallest value such that the product of an $N \times \lceil N^{k}\rceil$ matrix by an $\lceil N^k\rceil \times N$ can be computed using $O\big(N^{\rmm(k) + \epsilon}\big)$ arithmetic operations for any constant $\epsilon > 0$; for a formal definition of $\rmm(k)$, see Gall and Urrutia~\cite{Gall2018}. Thus, the exponent of square matrix multiplication is $\rmm := \rmm(1)$. 

We will make repeated use of the following facts about binomial coefficients:\begin{fact}\label{fact:binom}
For integers $k \geq 1$ and $n \geq 2k$ we have 
\bes
	(2n)^{-1/2} b\Big(\frac{k}{n}\Big)^n 
	\,\leq\, {n \choose k} 
	\,\leq\, \sum_{j=0}^k {n \choose j} 
	\,\leq\, b\Big(\frac{k}{n}\Big)^n = 2^{n H(k/n)}\,, 
\ees
where 
\bes 
	b(x) := x^{-x}(1-x)^{x-1}
	\quad \textrm{and} \quad
	H(x) :=	\log_2 b(x)
	\,,\qquad x \in [0, 1]\,. 
\ees
\end{fact}
This can be proven using Stirling's approximation to factorials. 
\begin{fact}\label{fact:binommax}
Let $n$ be a positive integer. The function $k \mapsto {n \choose k} 2^k$ is increasing in $[0, \tfrac{2}{3}n)$ and strictly decreasing in $[\tfrac{2}{3}n, n)$. 
\end{fact}
This can be proven by observing that the ratio $\binom{n}{k+1} 2^{k+1} / \binom{n}{k} 2^{k}$ equals $2 (n-k)/(k+1)$, and is thus decreasing in $k$, and is greater or equal to $1$ exactly when $k \leq \tfrac{2}{3}n - \tfrac{1}{3}$. 

\subsection{Basic reduction to rectangular matrix multiplication}
\label{se:basic}

Assume without loss of generality that $n$ is even. Let us arbitrarily partition $U$ into two disjoint sets $U_1$ and $U_2$, both of size $h := n/2$. If $T \subseteq U$, denote by $T_1$ and $T_2$ respectively the intersections $T\cap U_1$ and $T\cap U_2$. Furthermore, write $N := 2^h$ so that $2^n = N^2$. 

Armed with this notation, we write the multi-subset transform of set functions $(f_i)_{i\in U}$ as  
\be \label{eq:rmm}
    g(T) = G(T_1, T_2) := \sum_{S \subseteq U} F_1(T_1, S)\, F_2(T_2, S)\,, 
    \qquad T \subseteq U\,,
\ee
where we define  
\bes
    F_p(T_p, S) := [S \cap U_p \subseteq T_p] \prod_{i \in T_p} f_i(S)\,, 
    \qquad p = 1, 2\,.
\ees
Here the Iverson's bracket notation $[Q]$ evaluates to $1$ if $Q$ is true, and to $0$ otherwise. 

We can write the representation \eqref{eq:rmm} in terms of a matrix product as 
\bes
    G = F_1 F^{\top}_2\,,
\ees
where $G$ is an $N \times N$ matrix indexed in $2^{U_1} \times 2^{U_2}$ and $F_p$ is an $N \times N^2$ matrix indexed in $2^{U_p} \times 2^{U}$. As above, we will write the index pair in parentheses (not as subscripts).  

Applying fast RMM without any further tricks already yields a somewhat competitive asymptotic complexity bound. To see this, recall that $\rmm(k)$ denotes the exponent of RMM of dimensions $N \times \lceil N^{k}\rceil \times N$. Since $\rmm(2) < 3.252$ \cite{Gall2018}, we get that $G$, and thus $g$, can be computed using $O\big(N^{3.252}\big) = O(3.087^n)$ arithmetic operations. If the lower bound $\rmm(2) \geq 3$ was tight, we would achieve the bound $O(2.829^n)$.

So far, we have ignored the sparsity of the matrices $F_p$. An entry $F_p(T_p, S)$ is zero whenever the intersection $S_p = S \cap U_p$ is not contained in $T_p$. Thus, out of the $8^{n/2}$ entries of $F_p$, at most $3^h 2^h = 6^{n/2}$ are nonzero. In general, one can compute a matrix product of dimensions $r \times r^k \times r$ using $O\big(m r^{(\rmm-1)/2+\epsilon}\big)$ operations, provided that the matrices have at most $m \geq r^{(\rmm+1)/2}$ non-zero entries, irrespective of $k$ \cite{Kaplan2006}. This result applies to our case, but with the best known upper bound for $\rmm$ \cite{Gall2014}, it only yields a bound $O(3.108^n)$. A direct reduction to multiple multiplications of sparse square matrices \cite{Yuster2005} yields an even worse bound, $O(3.142^n)$ (calculations omitted). Output-sensitive sparse matrix multiplication algorithms  \cite{Amossen2009} will not work either, as our output matrix is dense in general.

Luckily, in our case, we can make more efficient use of the sparsity. We will decompose the matrix product into a sum of smaller matrix products,
as formulated by the following representation (the proof is trivial and omitted):  
\begin{lemma}\label{lem:split}
Let $\{\cS_1, \cS_2, \ldots, \cS_M\}$ be a set partition of $2^U$. Let $F_{pq}$ be the submatrix of $F_p$ obtained by removing all columns but those in $\cS_q$, for $p = 1, 2$ and $q = 1, 2, \ldots, M$. Then 
\bes
	G = \sum_{q=1}^M G_q\,,\quad \textrm{ where } G_q = F_{1q} F_{2q}^{\top}\,. 
\ees 
\end{lemma}

We will also apply this decomposition after removing some rows from the matrices $F_{pq}$. Then the index sets may be different for different $G_q$. To properly define the entry-wise addition in these cases, we simply make the convention that the missing entries equal zero.  

To employ a fast RMM algorithm we will call a function $\proc{Fast-RMM}(\cT_1, \cS, \cT_2)$. The function returns the product $E_1 E_2^{\top}$, where each $E_p$ is obtained from $F_p$ by only keeping the rows $\cT_p$ and the columns $\cS$. Note that we do not show the input matrices explicitly in the function call, as the submatrices will always be extracted from $F_1$ and $F_2$.

\subsection{A simple below-$3$ algorithm} 
\label{se:first}

We apply Lemma~\ref{lem:split} with $M = 2$ and split the columns to those that are smaller than a threshold $\sigma n$ and to those that are at least as large: 
\bes
	\cS_1 = \{S \subseteq U : |S| < \sigma n \}
	\quad \textrm{and} \quad
	\cS_2 = \{S \subseteq U : |S| \geq \sigma n \}\,.
\ees
We assume $\sigma n$ is an integer and that $\tfrac{1}{3} < \sigma < \tfrac{1}{2}$. We will optimize the parameter $\sigma$ later. 
The idea is to call fast RMM only for summing over the columns $\cS_1$ and to handle the remaining columns in a brute-force manner. The algorithm $\proc{Column}$ is given in Figure~\ref{fig:Columns}.  

\begin{figure}[t!]
 \begin{codebox}
 \Procname{{\bf function} $\proc{Columns-Directly}(\cS)$}
 \li $G[T] \gets 0$ for all $T \subseteq U$  
 \li \For $S \in \cS$
 \li \Do \For $T \subseteq U$ s.t.\ $S \subseteq T$
 \li \Do $G[T] \gets G[T] + F_1(T_1, S) F_2(T_2, S)$\End\End
 \li \Return $G$	
 \end{codebox}
 \begin{codebox}
 \Procname{{\bf Algorithm} $\proc{Columns}\big((f_i)_{i \in U}\big)$}
 \li $G[T] \gets 0$ for all $T \subseteq U$ 
 \li select $\sigma \in (\tfrac{1}{3}, \tfrac{1}{2})$
 \li $\cS_1 \gets \{S \subseteq U : |S| \leq \sigma n\}$
 \li $G \gets G + \proc{Fast-RMM}\big(2^{U_1}, \cS_1, 2^{U_2}\big)$
 \li $G \gets G + \proc{Columns-Directly}\big(2^U \setminus \cS_1\big)$
 \li \Return $G$
 \end{codebox}
\caption{The $\proc{Columns}$ algorithm for the multi-subset transform.}
\label{fig:Columns}
\end{figure}

Consider first the computation of the matrix $G_1$. We compute $G_1$ using fast RMM. The computational complexity depends on the number of columns in the matrices $F_{11}$ and $F_{21}$. Letting $C$ be the number of columns, the required number of operations for the matrix multiplication of dimensions $N \times C \times N$ is $O\big(N^{\rmm(k)}\big)$, where $k = \log_N C$. We have   
\be\label{eq:C}
	C = |\cS_1| = \sum_{s=0}^{\sigma n} {n \choose s} 
	\,\leq\, b(\sigma)^n\,,
\ee
where the inequality follows by Fact~\ref{fact:binom}.

Consider then the computation of the matrix $G_2$. To compute $G_2(T)$, for $T \subseteq U$, it suffices to compute the sum of the products $F_1(T_1, S) F_2(T_2, S)$ over all columns $S \subseteq T$  whose size is at least $\sigma n$. Thus, the required number pairs $(S, T)$ to be considered is at most
\be\label{eq:B}
	B := \sum_{s=\sigma n}^n {n \choose s} 2^{n-s}
	\,\leq\, n {n \choose \sigma n} 2^{n(1-\sigma)}
	\,\leq\, n\big(2^{1 - \sigma} b(\sigma)\big)^n	
\ee
where the penultimate inequality follows by Fact~\ref{fact:binommax} (since $1- \sigma < \tfrac{2}{3}$) and the last by Fact~\ref{fact:binom}. 

Let us finally combine the bounds in \eqref{eq:C} and \eqref{eq:B}.  

\begin{proposition}
For any $\epsilon > 0$, the number of operations required by $\proc{Columns}$ is 
\bes
	O\left(2^{n \left(\rmm\left(2 H(\sigma)\right) +\epsilon\right)/2} 
	+ n 2^{n\left(1-\sigma + H(\sigma)\right)} \right)\,.
\ees 
\end{proposition}

It remains to choose $\sigma$ so as to optimize the bound. Clearly the first term is increasing and the second term is decreasing in $\sigma$. Thus, the bound is (asymptotically) minimized by choosing a $\sigma$ that makes $\rmm\big(2 H(\sigma)\big)$ equal to  $2(1 - \sigma + H(\sigma))$. 
There are two obstacles to implement this idea: first, we only know upper bounds for $\rmm(k)$, for various $k$; second, no closed-form expression is known for the best upper bounds---upper bounds for $\rmm(k)$ have been computed and reported only at some points $k$ \cite{Gall2018}. 

Due to these complications, we resort to the following facts: 
\begin{fact}[\cite{Gall2018}]\label{fact:rmm175}
The exponent of RMM satisfies $\rmm(1.75) \leq 3.021591$. 
\end{fact}
\begin{fact}\label{fact:rmmbound}
Let $k > 0$ and $r \geq 0$. The exponent of RMM satisfies $\rmm(k + r) \leq \rmm(k) + r$.
\end{fact}
(This follows by reducing the larger RMM instance trivially to multiple smaller instances.) 

Combining these two facts yields an upper bound:  
\bes
	\rmm(2 H(\sigma)) \leq \rmm(1.75) + 2 H(\sigma) - 1.75 \leq 1.271591 + 2 
H(\sigma)\,.
\ees
Now, solving 
	$1.271591 + 2 H(\sigma) = 2(1 - \sigma + H(\sigma))$
gives 
\bes
	\sigma = 1 - 1.271591/2 = 0.3642045\,.
\ees
With this choice of $\sigma$ the complexity bound becomes $O(2.994^n)$.

\subsection{A faster below-$3$ algorithm} 
\label{se:second}

Next we give a slightly faster algorithm to compute $G_1$. This will allow us to choose a larger threshold $\sigma$, thus also rendering the computation of $G_2$ faster.    

Instead of computing $G_1$ directly using fast RMM, we now compute some rows and columns of $G_1$ in a brute-force manner and only apply fast RMM to the remaining smaller matrix. Specifically, the algorithm only calls fast RMM to compute the entries $G_1(T_1, T_2)$ where the sizes of $T_1$ and $T_2$ exceed $\tau h$. We assume that $\tau h$ is an integer and that $\tau \in (\tfrac{1}{2}, \tfrac{2}{3})$. We will optimize the parameter $\tau$ together with $\sigma$ later. The algorithm $\proc{Rows{\&}Columns}$ is given in Figure~\ref{fig:Rows}. The correctness of the algorithm being clear, we proceed to analysing the  complexity in terms of the required number of arithmetic operations. 

\begin{figure}[t!]
 \begin{codebox}
 \Procname{{\bf function} $\proc{Rows-Trimmed}(\tau, \cS)$}
 \li $G[T] \gets 0$ for all $T \subseteq U$ 
 \li $\cT_p \gets \{T_p \subseteq U_p : |T_p| > \tau h\}$ for $p \gets 1, 2$
 \li \For $S \subseteq T \subseteq U$ s.t.\ $S\in \cS$ and ($T_1 \not\in \cT_1$ or $T_2 \not\in \cT_2$)
 \li \Do $G[T] \gets G[T] + F_1(T_1, S) F_2(T_2, S)$\End
 \li $G \gets G + \proc{Fast-RMM}\big(\cT_1, \cS, \cT_2\big)$
 \li \Return $G$
 \end{codebox}
 \begin{codebox}
 \Procname{{\bf Algorithm} $\proc{Rows{\&}Columns}\big((f_i)_{i \in U}\big)$}
 \li $G[T] \gets 0$ for all $T \subseteq U$ 
 \li select $\sigma \in (\tfrac{1}{3}, \tfrac{1}{2})$ and $\tau \in (\tfrac{1}{2}, \tfrac{2}{3})$
 \li $\cS_1 \gets \{S \in U : |S| \leq \sigma n\}$
 \li $G \gets G + \proc{Rows-Trimmed}(\tau, \cS_1)$
 \li $G \gets G + \proc{Columns-Directly}\big(2^U \setminus \cS_1\big)$
 \li \Return $G$
 \end{codebox}
\caption{The $\proc{Rows{\&}Columns}$ algorithm for the multi-subset transform.}
\label{fig:Rows}
\end{figure}

Consider first the computation of an entry $G_1(T_1, T_2)$ where $|T_1| \leq \tau h$. The number of pairs $(S, T)$ satisfying $S \subseteq T \subseteq U$ and $|T_1| \leq \tau h$ is given by 
\be \label{eq:Bprime}
	B' := 3^h \sum_{t=0}^{\tau h} {h \choose t} 2^t 
	\,\leq\, 3^h h {h \choose \tau h} 2^{\tau h} 
	\,\leq \,h \big(3 \cdot 2^\tau b(\tau)\big)^h\,; 
\ee
the penultimate inequality follows by Fact~\ref{fact:binommax} (since $\tau < \tfrac{2}{3}$) and the last inequality by Fact~\ref{fact:binom}. 

Similarly, computing the entries $G_1(T_1, T_2)$ for all $T_1 \subseteq U_1$ and $T_2 \subseteq U_2$ such that $|T_2| \leq \tau h$ requires at most $B'$ additions and multiplications. 

It remains to compute the entries $G_1(T_1, T_2)$ for $T_1 \subseteq U_1$ and $T_2 \subseteq U_2$ such that $|T_1|, |T_2| > \tau h$. This can be computed as a product of two matrices (submatrices of $F_1$ and $F_2^ \top$) whose sizes are at most $R \times C$ and $C \times R$, where $C$ is as before and 
\be\label{eq:R}
	R := \sum_{j=\tau h + 1}^{h}{h \choose j} \,\leq\, b(\tau)^h\,,
\ee
where the inequality follows by Fact~\ref{fact:binom} (since $\tau > \tfrac{1}{2}$).

Let us combine the bounds in \eqref{eq:Bprime} and \eqref{eq:R}: 
\begin{proposition}
For any $\epsilon > 0$, the number of operations required by $\proc{Rows{\&}Columns}$ is 
\bes
	O\Big(
	n (3 \cdot 2^{\tau} b(\tau))^{n/2}	
	+ b(\tau)^{(\rmm(k) +\epsilon)n/2} 
	+ n \big(2^{1-\sigma} b(\sigma)\big)^n 
	\Big)\,,
	\quad \textrm{where }
	k = 2 \log_{b(\tau)} b(\sigma)\,.
\ees
\end{proposition}

To set the parameters $\sigma$ and $\tau$, we resort to the bound $\rmm(k) \leq 1.271591 + k$ (Fact~\ref{fact:rmm175} and Fact~\ref{fact:rmmbound}). Balancing the latter two terms in the bound yields the equation  
\bes
	(1.271591 + k) H(\tau)
	= 
	2\big(1 - \sigma + H(\sigma)\big)\,.
\ees
Equivalently, $1.271591 \cdot H(\tau) = 2(1 - \sigma)$. Solving for $\sigma$ and equating the first and the third term in the bound leaves us the equation
\bes
	\log_2 3 + \tau + H(\tau) 
	= 
	1.271591 \cdot H(\tau) + 2 H\big(1 - 0.6357955 \cdot H(\tau)\big)\,.
\ees
By numerical calculations we find one solution in the valid range, $\tau \approx 0.59777$, and correspondingly $\sigma \approx 0.38185$.  
With these choices the complexity bound becomes $O(2.985^n)$. 
This completes the proof of Theorem~\ref{thm:fmst}.


\subsection{A covering based algorithm} 
\label{se:third}

The previous algorithms were based on pruning some columns and rows of the matrices $F_1$ and $F_2$, and applying fast RMM to the remaining multiplication of two reduced matrices. Now, we take a different approach and reduce the original problem instance into multiple, smaller RMM instances applying Lemma~\ref{lem:split} with some $M > 2$. To this end, we cover---in the sense of a set cover---the columns by multiple groups such that the columns in one group contain a large block of zero entries (in the same set of rows) in the matrices $F_1$ and $F_2$.  

It will be convenient to consider sets of fixed sizes. For a set $V$ and a nonnegative integer $s$, write ${V \choose s}$ for the set of all $s$-element subsets of $V$. Let $s_1, s_2 \in \{0, 1, \ldots, h\}$ fix the sizes of the intersection of a column with the sets $U_1$ and $U_2$. We wish to cover the set (of set pairs) ${U_1 \choose s_1} \times {U_2 \choose s_2}$ by a small number of sets of the form ${K_1 \choose s_1} \times {K_2 \choose s_2}$, where the sets $K_1$ and $K_2$ are of some fixed sizes $k_1 \geq s_1$ and $k_2 \geq s_2$. The following classic result \cite{Erdos1974} shows that this covering design problem has an efficient solution: 

\begin{theorem}[\cite{Erdos1974}]\label{thm:cover}
Let $c(v, k, s)$ be the minimum number of subsets of $\{1,2,\ldots,v\}$ of size $k$ such that every subset of size $s \leq k$ is contained by at least one of the sets. We have 
\bes
    c(v, k, s) {k \choose s}{v \choose s}^{-1} \leq 1 + \ln {k \choose s}\,.
\ees
\end{theorem}
In particular, $c(v, k, s)$ is within the factor $k$ of the obvious lower bound ${v \choose s}{k \choose s}^{-1}$. 

\begin{remark}
Although the work needed for constructing a covering does not contribute to the number of operations in the ring $\cR$, a remark is in order if one is interested in the required number of other operations. The authors are not aware of any deterministic algorithm for constructing an optimal covering in time polynomial in ${v \choose k} + {v \choose s}$, while asymptotically optimal randomized polynomial-time algorithms are known \cite{Gordon1996}. 

Fortunately, for our purposes it suffices to run the well known greedy algorithm that iteratively picks a set that covers the largest number of yet uncovered elements. It finds a set cover whose size is within a logarithmic factor of the optimum, which is sufficient in our context. Furthermore, it can be implemented to run in time linear in the input size \cite[Ex.~35.3--3]{Cormen2009}, which is ${v \choose k} {k \choose s} \leq 3^v$ in our case (with $v = h = n/2$). 
\end{remark}

From now on, we assume that for $p = 1, 2$ we are given a set family $\cK_p \subseteq {U_p \choose k_p}$ that has the desired coverage property, i.e., $\big\{ {K_p \choose s_p} : K_p \in \cK_p \big\}$ is a set  cover of ${U_p \choose s_p}$, so that for every column $S \subseteq U$ satisfying $|S_1| = s_1$, $|S_2| = s_2$ there is a pair $(K_1, K_2) \in \cK_1 \times \cK_2$ such that $S_1 \subseteq K_1$, $S_2 \subseteq K_2$. In what follows, we will assume that some appropriate values of $k_1$, $k_2$ are chosen based on $s_1, s_2$; we will return back to the issue of finding good values at the end of this subsection. 

For each pair $(K_1, K_2)$, we construct a submatrix $E_1$ of $F_1$ as follows: remove from $F_1$ all columns $S$ not covered by $(K_1, K_2)$, and all rows $T_1$ whose intersection with $K_1$ contains less than $s_1$ elements (as otherwise we cannot have $S_1 \subseteq T_1$ and the entry $F_1(T_1, S)$ vanishes). We construct a matrix $E_2$ analogously by removing columns and rows from $F_2$. The dimensions of the matrix product $E_1 E_2^{\top}$ are $R_1 \times C' \times R_2$, where 
\bes
    R_1 := \sum_{j=s_1}^{k_1} {k_1 \choose j} 2^{h-k_1}\,,\quad 
    C' := {k_1 \choose s_1}{k_2 \choose s_2}\,,\quad 
    R_2 := \sum_{j=s_2}^{k_2} {k_2 \choose j} 2^{h-k_2}\,.
\ees

Algorithm $\proc{Cover-Columns}$, given in Figure~\ref{fig:cover}, organizes the reduction to multiple RMM instances like this using Lemma~\ref{lem:split}. Specifically, from the set cover of the columns it extracts a set partition by trivially keeping track of the already covered columns.   

\begin{figure}[t!]
 \begin{codebox}
 \Procname{{\bf Algorithm} $\proc{Cover-Columns}\big((f_i)_{i\in U}\big)$}
 \li $G[T] \gets 0$ for all $T \subseteq U$ 
 \li $\cC \gets \emptyset$ \RComment Already covered columns
 \li \For $(s_1, s_2) \in \{0, 1, \ldots, h\}^2$
 \li \Do select $k_1$ and $k_2$
 \li $\cK_p \gets \proc{Covering-Design}(s_p, k_p, U_p)$ for $p \gets 1, 2$
 \li \For $(K_1, K_2) \in \cK_1 \times \cK_2$ 
 \li \Do $\cS \gets \{S_1 \cup S_2 : S_1 \in K_1 \text{ and } S_2 \in K_2\}$ 
 \li $G \gets G + \proc{Rows-Trimmed}(0, \cS \setminus \cC)$ \RComment Trim only all-zero rows 
 \li $\cC \gets \cC \cup \cS$ \End \End
 \li \Return $G$
 \end{codebox}
\caption{The $\proc{Cover-Columns}$ algorithm for the multi-subset transform.}
\label{fig:cover}
\end{figure}

To analyze the complexity of the algorithm, let us first bound the dimensions $R_1$, $C'$, and $R_2$ for fixed $s_1$, $s_2$, $k_1$, $k_2$. We aim at bounds of the form $N^{\alpha}$ for some $0 < \alpha < 2$, and therefore parameterize the set sizes as 
\bes
	s_p = \sigma_p h 
	\quad \textrm{and} \quad 
	k_p = \kappa_p h\,, \qquad p = 1, 2\,.
\ees
Thus $0 \leq \sigma_p \leq \kappa_p \leq 1$. In what follows, we let $\sigma_p/\kappa_p$ evaluate to $0$ if $\sigma_p = \kappa_p = 0$. 

\begin{lemma}\label{lem:aux}
We have 
\bes
    R_1 \,\leq\, N^{\beta_1}\,,\quad 
    C' \,\leq\, N^{\alpha_1 + \alpha_2}\,,\quad 
    R_2 \,\leq\, N^{\beta_2}\,, 
\ees
where 
\be \label{eq:aux}
	\alpha_p := \kappa_p H\Big(\frac{\sigma_p}{\kappa_p}\Big) 
	\quad \textrm{and} \quad
	\beta_p := 1 - \kappa_p + \kappa_p H\Big(\max\Big\{\frac{\sigma_p}{\kappa_p}, \frac{1}{2}\Big\}\Big)\,,  \qquad p = 1, 2\,.
\ee
\end{lemma}
\begin{proof}
The bound for $C'$ follows directly from the definitions of $\sigma_p$, $\kappa_p$, $\alpha_p$ and from Fact~\ref{fact:binom}. 

For the bound on $R_1$ (equivalently $R_2$), suppose first that $\kappa_1 \geq 2 \sigma_1$. Then using the simple inequality $\sum_{j=s_1}^{k_1} {k_1 \choose j} \leq 2^{k_1} = N^{\kappa_1 H(1/2)}$ gives the claimed bound. Otherwise, $\kappa_1 \leq 2 \sigma_1$ and thus, by Fact~\ref{fact:binom}, $\sum_{j=s_1}^{k_1} {k_1 \choose j} \leq 2^{k_1 H(1-\sigma_1/\kappa_1)} = N^{\kappa_1 H(\sigma_1/\kappa_1)}$, implying the claimed bound. 
\end{proof}

It remains to turn the bounds on the dimensions to a bound on the complexity of the corresponding RMM and sum up these bounds over the multiple matrix multiplication tasks.
\begin{proposition}\label{prop:cc}
For any $\epsilon > 0$, the number of operations required by $\proc{Cover-Columns}$ is $O\big( 2^{(\gamma+\epsilon) n/2}\big)$, where 
\be \label{eq:gamma}  
	\gamma := 	
	\max_{\substack{0 \leq \sigma_1 \leq 1\\
		0 \leq \sigma_2 \leq 1}}	
	\min_{\substack{\sigma_1 \leq \kappa_1 \leq 1\\
		\sigma_2 \leq \kappa_2 \leq 1}}
	H(\sigma_1) + H(\sigma_2) - \alpha_1 - \alpha_2 
	+ \beta_1 + \beta_2
	+  \beta_{*} \Big(\rmm\Big(\frac{\alpha_1 + \alpha_2}{\beta_{*}}\Big) - 2\Big)\,,
\ee   
with $\alpha_p$ and $\beta_p$ as defined in \eqref{eq:aux}, and $\beta_{*} := \min\{\beta_1, \beta_2\}$.
\end{proposition}
\begin{proof}
Let $\epsilon > 0$. 

Consider first the complexity of a single matrix multiplication with fixed $\sigma_p, \kappa_p$, for $p = 1, 2$. By Lemma~\ref{lem:aux} we obtain an upper bound by taking $N^{\max\{\beta_1, \beta_2\} - \beta_{*}} = N^{\beta_1 + \beta_2 - 2 \beta_{*}}$ matrix multiplications of dimensions $N^{\beta_{*}} \times N^{\alpha_1 + \alpha_2} \times N^{\beta_{*}}$. This gives us the upper bound $O\big(N^{\beta_1 + \beta_2 + \beta_{*}(\rmm(k) - 2) + \epsilon/2}\big)$, where $k = (\alpha_1 + \alpha_2)/\beta_{*}$. Note that we used only a half of $\epsilon\,$---we will need the other half for tolerating a nonzero underestimation that is due to minimizing $\kappa_p$ over reals. We will return to this issue at the end of the proof.  

Consider then the number of matrix multiplications for fixed $s_p, k_p$, for $p = 1, 2$. By Theorem~\ref{thm:cover} and by the approximation ratio of the greedy algorithm, the number is at most 
\bes
	n^4 {h \choose s_1}{h \choose s_2}{k_1 \choose s_1}^{-1}{k_2 \choose s_2}^{-1} 
	&\leq & n^5 b(\sigma_1)^h b(\sigma_2)^h b(\sigma_1/\kappa_1)^{-\kappa_1 h} b(\sigma_2/\kappa_2)^{-\kappa_2 h}\\
	& = & n^5 N^{H(\sigma_1)+H(\sigma_2) - \alpha_1 - \alpha_2}\,.
\ees
Here we used Fact~\ref{fact:binom} to bound the binomial coefficients, observing that $(2k_1)^{1/2}(2k_2)^{1/2} \leq n$. 

Now, combine the above two bounds, recall that $N = 2^{n/2}$, and observe that replacing the sum over $(s_1, s_2)$ by the maximum over $(\sigma_1, \sigma_2)$ is compensated by adding a factor of $n^2$ to the bound. The algorithm can select optimal $k_1$ and $k_2$ by optimizing the upper bound, which costs yet another factor of $n^2$. Due to the constant $\epsilon$ in the exponent, we can ignore the $n^{O(1)}$ factor in the asymptotic complexity bound.

To complete the proof, we show that for any values of $\sigma_p$ and $\kappa_p$ (hence also for the optimal values) and for any large enough integer $h$, there are rational numbers $\kappa'_p \geq \sigma_p$ such that 
(i) $\kappa'_p h$ are integers and 
(ii) $\Gamma(\sigma_1, \sigma_2, \kappa'_1, \kappa'_2) \leq \Gamma(\sigma_1, \sigma_2, \kappa_1, \kappa_2) + \epsilon/2$, where 
\be \label{eq:Gammaaux}
	\Gamma(\sigma_1, \sigma_2, \kappa_1, \kappa_2) 
	:= H(\sigma_1) + H(\sigma_2) + \beta_1 + \beta_2 
	+ \beta_{*}\Big(\rmm\Big(\frac{\alpha_1 + \alpha_2}{\beta_{*}}\Big)  - \Big(\frac{\alpha_1 + \alpha_2}{\beta_{*}}\Big) - 2\Big)\,.
\ee 
Note that we rearranged some terms in \eqref{eq:gamma}, for a reason that will be revealed in a moment.

We will consider two cases: either $\sigma_1$ or $\sigma_2$ is near the boundary values $0$ or $1$, or both are in $[c, 1-c]$, where $c > 0$ is a small constant. We choose $c < \frac{1}{2}$ such that if $0 \leq \sigma_1 < c$ or $1-c < \sigma_1 \leq 1$, then regardless of $\sigma_2$,  
\bes
	\Gamma(\sigma_1, \sigma_2, 1, 1)
	\leq \rmm(1) + \epsilon/2\,,  
\ees
and symmetrically for $\sigma_2$. To see that this is possible, observe first that at $\kappa_1 = \kappa_2 = 1$ we have $\alpha_1 = H(\sigma_1)$, $\alpha_2 = H(\sigma_2)$, and thus   
\bes
	\Gamma(\sigma_1, \sigma_2, 1, 1)
	&=& \beta_1 + \beta_2 
	+ \beta_{*}\Big(\rmm\Big(\frac{\alpha_1 + \alpha_2}{\beta_{*}}\Big)  - 2\Big)\\
	&\leq& \beta_1 + \beta_2 
	+ \beta_{*}\Big(\rmm\Big(\frac{\alpha_{*}}{\beta_{*}}\Big) + \frac{\alpha_1 + \alpha_2 - \alpha_{*}}{\beta_{*}}  - 2\Big)\,, 
\ees
where $\alpha_{*} := \alpha_p$ if $\beta_{*} = \beta_p$. Observe that $\alpha_{*} \leq \beta_{*}$. Since $\rmm(1) - 2 \geq 0$ and $\alpha_1, \alpha_2, \beta_1, \beta_2 \leq 1$, 
\bes
	\Gamma(\sigma_1, \sigma_2, 1, 1)
	\;\leq\; \alpha_1 + \alpha_2 - \alpha_{*} + \beta_1 + \beta_2 + \rmm(1) - 2
	\;\leq\; \rmm(1) + H(\sigma_1)\,.
\ees
For the latter inequality we used the facts that $\alpha_{*} = \alpha_2$ if $\sigma_1 < c$ and that $\beta_1 = H(\sigma_1)$ if $\sigma_1 > 1 - c$. 
Finally, we observe that $H(\sigma_1)$ tends to $0$ when $\sigma_1$ tends to $0$ or $1$. 

On the other hand, we have the lower bound  
$\Gamma(\frac{1}{2}, \frac{1}{2}, \kappa_1, \kappa_2)
	\geq 2 + \beta_1 + \beta_2 - \beta_{*} 
	\geq 2.5 > \rmm(1)$, 
since $\rmm(z) - z \geq 1$ and $\beta_p = 1 - \kappa_p + \kappa_p H\big(1/(2\kappa_p)\big)\geq \kappa_p \geq \frac{1}{2}$; here we used the fact that $H(x) \geq 2 - 2x$ for $x \in \big[\frac{1}{2}, 1\big]$.  

We may thus restrict out attention to the domain 
\bes
	\Lambda_c := \big\{(\sigma_1, \sigma_2, \kappa_1, \kappa_2) : 
	c \leq \sigma_1, \sigma_2 \leq 1-c,\; 
	\sigma_1 \leq \kappa_1 \leq 1,\; 
	\sigma_2 \leq \kappa_2 \leq 1 \big\}\,.
\ees

We now show that $\Gamma$ is continuous on $\Lambda_c$. Observe first that the functions $H$, $\alpha_p$, and $\beta_p$ are continuous on $\Lambda_c$ (as $\kappa_p > c$). We also have that $\beta_{*}$ is continuous and strictly positive (as $\sigma_p \leq 1-c$) and that $z \mapsto \rmm(z)$ is continuous (as $|\rmm(z+\delta) - \rmm(z)| \leq \delta$ for all $\delta > 0$).   

Since the domain $\Lambda_c$ is compact, we have that $\Gamma$ is uniformly continuous on $\Lambda_c$. This in turn implies that there is a $\delta_{\epsilon} > 0$ such that (ii) holds whenever $|\kappa'_p - \kappa_p| < \delta_{\epsilon}$, implying that we can make both (i) and (ii) hold for all $h > 1/\delta_{\epsilon}$ by putting $\kappa'_p := \lceil \kappa_p h \rceil / h$.  
\end{proof}

Now we know that the complexity of the algorithm is $O\big( 2^{(\gamma+\epsilon) n/2}\big)$, but we do not know how large $\gamma$ is. 
Unlike for the simpler algorithms given in the previous subsections, we cannot just select some values of the parameters $\sigma_p$ and $\kappa_p$ and bound $\gamma$ from above by $\Gamma(\sigma_1, \sigma_2, \kappa_1, \kappa_2)$, as defined in \eqref{eq:Gammaaux}, for we do not know the maximizing values of $\sigma_p$. Since $\Gamma$ is uniformly continuous on the domain $\Lambda_c$, one could in principle prove any fixed strict upper bound on $\gamma$ with a sufficiently large, finite computation.
While at the present time the authors have not produced such a proof, evaluations of $\Gamma(\sigma_1, \sigma_2, \kappa_1, \kappa_2)$ at various values of the four parameters suggest the following: 

\begin{conjecture}
The number of operations required by $\proc{Cover-Columns}$ is $O(2.930^n)$. 
\end{conjecture}

\section{Fast weighted counting of acyclic digraphs: proof of Theorem~2}

Let us write the inclusion--exclusion recurrence \eqref{eq:aV} as a multi-subset transform: 
\begin{lemma}
Without loss of generality, suppose $0 \not\in V$. 
Let $0 \in T \subseteq V\cup\{0\}$ and  
\bes
	g(T) = \sum_{S \subseteq T} \prod_{i\in T} f_i(S)\,,
\ees
where 
\bes
	f_i(S) = \left\{\begin{array}{ll}
	0 & \textrm{if $0 \not\in S$ or $|S| = |T|$};\\	
	(-1)^{|S|-1} a_{S\setminus\{0\}} & \textrm{else if $i = 0$};\\
	\sum_{D_i \subseteq S\setminus\{0\}} w_i(D_i) & \textrm{else if $i \not\in S$};\\
	1 & \textrm{otherwise}.
 	\end{array}\right.
\ees
Then $a_{T\setminus\{0\}} = (-1)^{|T|} g(T)$.
\end{lemma}
\begin{proof}
Because the summand vanishes unless $0 \in S \neq T$ and because $f_i(S) = 1$ unless $i \in \{0\}\cup (T \setminus S)$, we have
\bes
	(-1)^{|T|} g(T) &=& (-1)^{|T|} \sum_{0 \in S \subsetneq T} 
	f_0(S) \prod_{i \in T\setminus S} f_i(S)\\ 
	&=& \sum_{0 \in S \subsetneq T} 
	(-1)^{|T|+|S|-1} a_{S\setminus\{0\}} \prod_{i \in T\setminus S} 		\,\sum_{D_i \subseteq S\setminus\{0\}} w_i(D_i)\,.
\ees
Writing in terms of $T' := T\setminus\{0\}$ and $S' := T\setminus S$, and observing that $|S|$ and $-|S|$ have the same parity, 
\bes
	(-1)^{|T|} g(T) \;=\; \sum_{\emptyset \neq S' \subseteq T'} 
	(-1)^{|S'|-1} a_{T'\setminus S'} \prod_{i \in S'} 		\,\sum_{D_i \subseteq T'\setminus S'} w_i(D_i)
	\;=\; a_{T'}\,.
\ees
The last equality follows immediately from \eqref{eq:aV}.
\end{proof}

It remains to organize the computations so that when computing $a_T$ for some $T \subseteq V$, the values $a_S$ have already been computed for all $S \subsetneq T$. To this end, we proceed in increasing order by $|T|$: for each $t = 1, 2, \ldots, n$ in this order we simultanously compute the values $a_T$ for all $T \in {V \choose t}$ by calling the fast multi-subset transform, as detailed in algorithm $\proc{Sum-Acyclic-Digraphs}$ given in Figure~\ref{fig:sum}. As we only need $n$ calls, the asymptotic complexity bound (with a rounded constant base of the exponential) remains valid. 

\begin{figure}[t!]
\begin{center}
 \begin{codebox}
 \Procname{{\bf Algorithm} $\proc{Sum-Acyclic-Digraphs}\big((w_i)_{i\in V}\big)$}
 \li $a[\emptyset] \gets 1$; $a[S] \gets 0$ for all $\emptyset \neq S \subseteq V$
 \li compute $f_i[S\cup\{0\}] \gets \sum_{X \subseteq S} w_i(X)$ for all $i \in V$, $S \subseteq V$ using fast zeta transform
 \li $f_i[S] \gets 0$ for all $i \in V$, $S \subseteq V$.
 \li $f_i[S\cup\{0\}] \gets 1$ for all $i \in S \subseteq V$.
 \li \For $t \gets 1 \To n$ 
 \li \Do \For $S \in {V \choose t-1}$ 
 \li \Do $f_0[S\cup\{0\}] \gets (-1)^{|S|-1} a[S]$ \End
 \li $g \gets \proc{Fast-Multi-Subset-Transform}\big((f_i)_{i\in V\cup\{0\}}\big)$
 \li \For $T \in {V \choose t}$
 \li \Do $a[T] \gets (-1)^{|T|+1} g[T\cup\{0\}]$\End\End
 \li \Return $a[V]$
 \end{codebox}
\end{center}
\caption{The $\proc{Sum-Acyclic-Digraphs}$ algorithm for the sum over acyclic digraphs with modular weights. $\proc{Fast-Multi-Subset-Transform}\big((f_i)_{i\in U}\big)$ returns the multi-subset transform of $(f_i)_{i\in U}$. \label{fig:sum}}
\end{figure}

\bibliography{fastmmstransform}

\end{document}